\documentclass[11pt]{article}

\usepackage{etoolbox}
  
\usepackage[letterpaper,margin=1in]{geometry}
\usepackage{times}
\usepackage{amsthm}
\usepackage[margin=25pt,font=small,labelfont=bf]{caption}

\usepackage{bbold}
\usepackage{amsmath}
\usepackage{amssymb}
\usepackage{pgffor}
\usepackage[inline]{enumitem}

\setlist[enumerate]{nosep,topsep=0.3em}
\setlist[enumerate,1]{label=(\roman*)}
\setlist[enumerate,2]{label=(\alph*)}

\setlist[itemize]{nosep,topsep=0.1em}

\usepackage{mathtools}

\usepackage{bm}
\usepackage{microtype}
\usepackage{xspace}
\usepackage{xfrac}

\usepackage{array}
\usepackage[table]{xcolor}
\usepackage{multirow}
\usepackage{makecell}
\usepackage{booktabs}

\newcolumntype{M}[1]{>{\centering\arraybackslash}m{#1}}

\newcolumntype{L}[1]
{>{\raggedright\let\newline\\\arraybackslash\hspace{0pt}}p{#1}}

\newcolumntype{C}[1]
{>{\centering\let\newline\\\arraybackslash\hspace{0pt}}p{#1}}

\newcolumntype{R}[1]
{>{\raggedleft\let\newline\\\arraybackslash\hspace{0pt}}p{#1}}

\newcolumntype{D}[1]
{>{\raggedright\let\newline\\\arraybackslash\hspace{0pt}}m{#1}}

\newcolumntype{E}[1]
{>{\centering\let\newline\\\arraybackslash\hspace{0pt}}m{#1}}

\newcolumntype{F}[1]
{>{\raggedleft\let\newline\\\arraybackslash\hspace{0pt}}m{#1}}

\usepackage{url}
\usepackage{hyperref}

\hypersetup{colorlinks, linkcolor=darkblue, citecolor=darkgreen, urlcolor=darkblue}

\usepackage{float}
\usepackage{graphicx}

\graphicspath{{graphics/}}

\usepackage{tikz}

\usetikzlibrary{calc}
\usetikzlibrary{decorations.pathreplacing}
\usetikzlibrary{decorations.pathmorphing}
\usetikzlibrary{shapes.geometric}
\usetikzlibrary{fadings,decorations,fit}

\usepackage[vlined,ruled,algo2e]{algorithm2e}

\usepackage[textsize=scriptsize]{todonotes}

\definecolor{darkblue}{rgb}{0,0,0.38}
\definecolor{darkred}{rgb}{0.6,0,0}
\definecolor{darkgreen}{rgb}{0.1,0.35,0}

\makeatletter
\newcommand{\labeltarget}[1]{\Hy@raisedlink{\hypertarget{#1}{}}}
\makeatother

\newtheorem{theorem}{Theorem}
\newtheorem{lemma}[theorem]{Lemma}

\newtheorem{definition}[theorem]{Definition}

\newtheorem{corollary}[theorem]{Corollary}

\newtheorem{observation}[theorem]{Observation}

\newlength\bxheight
\newcommand\TJdomRaw[1]{\ensuremath{P_{#1\textrm{-join}}^\uparrow}}
\newcommand\TJdom[1]{\raisebox{0pt}[\bxheight]{\TJdomRaw{#1}}}

\newcommand\HK{\mathrm{HK}}
\newcommand\PHK{P_{\mathrm{HK}}}
\newcommand\PST{P_{\mathrm{ST}}}

\newcommand\OPT{\mathrm{OPT}}

\DeclareMathOperator{\supp}{supp}

\DeclareMathOperator{\odd}{odd}
\newcommand{\symdiff}{\bigtriangleup}

\title{A $1.5$-Approximation for Path TSP}

\author{
Rico Zenklusen\thanks{
Department of Mathematics, ETH Zurich, Zurich, Switzerland.
Email: \href{mailto:ricoz@math.ethz.ch}
{ricoz@math.ethz.ch}.
Supported by the Swiss National Science Foundation grant
200021\_165866.
}
}

\date{}

\begin{document}

\maketitle

\begin{abstract}
We present a $1.5$-approximation for the Metric Path Traveling Salesman Problem (Path TSP). All recent improvements on Path TSP crucially exploit a structural property shown by An, Kleinberg, and Shmoys [Journal of the ACM, 2015], namely that narrow cuts with respect to a Held-Karp solution form a chain. We significantly deviate from these approaches by showing the benefit of dealing with larger $s$-$t$ cuts, even though they are much less structured. More precisely, we show that a variation of the dynamic programming idea recently introduced by Traub and Vygen [SODA, 2018] is versatile enough to deal with larger size cuts, by exploiting a seminal result of Karger on the number of near-minimum cuts. This avoids a recursive application of dynamic programming as used by Traub and Vygen, and leads to a considerably simpler algorithm avoiding an additional error term in the approximation guarantee. We match the still unbeaten $1.5$-approximation guarantee of Christofides' algorithm for TSP. Hence, any further progress on the approximability of Path TSP will also lead to an improvement for TSP.
 \end{abstract}

\section{Introduction}

The Metric Traveling Salesman Problem (TSP) together with many of its natural variants, like the Metric Path Traveling Salesman Problem (Path TSP), are fundamental and classical problems in Combinatorial Optimization with an enormous influence on the field. For TSP, we are given a complete undirected graph $G=(V,E)$ with nonnegative edge lengths $\ell\colon E\longrightarrow \mathbb{R}_{\geq 0}$ satisfying the triangle inequality, and the task is to find a shortest Hamiltonian tour. For Path TSP, the task is to find a shortest path between two given distinct vertices $s$ and $t$ that visits every vertex exactly once. Both problems being APX-hard, there has been extensive research on the development of approximation algorithms, which proved to be a highly nontrivial task. For TSP, the more than four decades old $1.5$-approximation of Christofides~\cite{christofides_1976_worst-case} remains unbeaten. However, for special cases of TSP, in particular for so-called \emph{graph metrics}, where the lengths $\ell$ correspond to shortest distances in a unit-length graph on the vertex set $V$, exciting progress has been achieved recently (see~\cite{oveisgharan_2011_randomized,momke_2016_removing,mucha_2014_approximation,sebo_2014_shorter} and references therein).
Interestingly, for Path TSP, the situation was somewhat similar until a few years ago. Beginning of the 90's, Hoogeveen~\cite{hoogeveen_1991_analysis} showed that a natural variant of Christofides' algorithm yields an approximation ratio of $\sfrac{5}{3}$. No improvement was found for over twenty years, until An, Kleinberg, and Shmoys~\cite{an_2012_improving,an_2015_improving} presented an elegant $\sfrac{(1+\sqrt{5})}{2}\approx 1.618$-approximation, by exploiting a new structural insight on so-called narrow cuts, namely that they form a chain.\footnote{A set family is a \emph{chain} if for any two sets $A,B$ in the family, either $A\subseteq B$ or $B\subseteq A$ holds.}
This was the beginning of a series of exciting improvements on the approximation factor for Path TSP (see Table~\ref{tab:prevApprox}), and which all build upon the crucial chain structure of narrow cuts. These developments culminated in a recent breakthrough result by Traub and Vygen~\cite{traub_2018_approaching}, who obtained, for any fixed $\epsilon >0$, a $(1.5+\epsilon)$-approximation through a recursive dynamic program.

\begin{table}[h]
\begin{center}
\rowcolors{2}{black!0}{black!5}
\begin{tabular}{L{5cm}l}
\multicolumn{1}{l}{Reference} & \multicolumn{1}{l}{Ratio}\\ \toprule
Hoogeveen~\cite{hoogeveen_1991_analysis} & $1.667$ \\
An, Kleinberg, and Shmoys~\cite{an_2015_improving} & $1.618$ \\
Seb\H{o}~\cite{sebo_2013_eight-fifth} & $1.6$ \\
Vygen~\cite{vygen_2016_reassembling} & $1.599$ \\
Gottschalk and Vygen~\cite{gottschalk_2016_better} & $1.566$ \\
Seb\H{o} and van Zuylen~\cite{sebo_2016_salesman} & $1.529$ \\
Traub and Vygen~\cite{traub_2018_approaching} & $1.5 + \epsilon$ \\
\end{tabular}
\end{center}
\caption{Previous approximation guarantees (rounded) for Path TSP.
\protect\footnotemark
}
\label{tab:prevApprox}
\end{table}

\footnotetext{Except for the last result listed in the table, all other results also imply an equivalent upper bound on the integrality gap of the Held-Karp relaxation, which is known to be at least $1.5$ (see, e.g.,~\cite{an_2015_improving}).}

On a high level, all of these results build upon the same original approach of Christofides for TSP, which is based on first obtaining a connected graph---through a spanning tree---and then performing parity correction of the degrees by adding additional edges. More precisely, for TSP, Christofides' algorithm starts with a minimum length spanning tree $T\subseteq E$, and then computes a shortest $\odd(T)$-join $J\subseteq E$, where $\odd(T)\subseteq V$ are all odd degree vertices with respect to $T$, and for $Q\subseteq V$ with $|Q|$ even, a $Q$-join is an edge set with odd degrees precisely at $Q$.
Thus, the multiset obtained by combining $J$ and $T$ has even degrees everywhere. It can thus be interpreted as a trail, and, by skipping vertices that have been visited already when going through the trail, it can be shortcut to a Hamiltonian tour. Due to metric lengths, the shortcutting cannot increase the length, thus leading to a Hamiltonian tour of length at most $\ell(T) + \ell(J)$, where $\ell(F)\coloneqq \sum_{e\in F}\ell(e)$ for any edge set $F\subseteq E$.
The fact that this leads to a $1.5$-approximation for TSP follows by deriving that $\ell(T)\leq \OPT$ and $\ell(J) \leq \sfrac{\OPT}{2}$, where $\OPT$ is the length of a shortest TSP solution.
The analogous procedure for Path TSP, i.e., first finding a minimum length spanning tree $T$ and then doing parity correction through a minimum length $(\odd(T)\symdiff \{s,t\}$)-join,\footnote{We use $A \symdiff B\coloneqq (A\setminus B) \cup (B\setminus A)$ for the \emph{symmetric difference} of two sets $A$ and $B$.} was analyzed by Hoogeveen~\cite{hoogeveen_1991_analysis}, who showed that it leads to a $\sfrac{5}{3}$-approximation (see Table~\ref{tab:prevApprox}). This is asymptotically tight in the sense hat there are instances where the approximation factor of the algorithm is arbitrarily close to $\sfrac{5}{3}$.

The reason why Christofides' algorithm is only a $\sfrac{5}{3}$-approximation for Path TSP, whereas it is a $\sfrac{3}{2}$-approximation for TSP, is at the heart of the recent improvements on Path TSP.
To better understand this discrepancy, and also to introduce our approach later on, it is helpful to analyze the performance of Christofides' algorithm for TSP and Path TSP, respectively, in terms of the well-known Held-Karp relaxation:
\begin{equation}\label{eq:HKforTSP}\tag{Held-Karp relaxation for TSP}
\renewcommand\arraystretch{1.2}
\begin{array}{r@{\;}c@{\;}ll}
\multicolumn{1}{c}{\min \ell(x)} & & & \\
     x(\delta(C)) &\geq & 2 & \forall\; C\subsetneq V,\; C\neq \emptyset\\
     x(\delta(v)) &= & 2 & \forall\; v\in V\\
     x &\in &\mathbb{R}^E_{\geq 0}\enspace,&
\end{array}
\end{equation}
where $\ell(x)\coloneqq \sum_{e\in E}x(e)\ell(e)$, the set $\delta(C)\subseteq E$ are all edges with precisely one endpoint in $C$ (and $\delta(v) \coloneqq \delta(\{v\})$), and $x(U)\coloneqq \sum_{e\in U}x(e)$ for any $U\subseteq E$.
Any feasible point $y$ to the Held-Karp relaxation for TSP can be shown to satisfy
\begin{enumerate*}
\item the scaled-down point $\frac{n-1}{n}\cdot y$ is in the spanning tree polytope of $G$, where $n\coloneqq |V|$, and
\item $\sfrac{y}{2}$ is in the dominant of the $Q$-join polytope of $G$ for any $Q\subseteq V$ of even cardinality, which can be described as follows (see~\cite{edmonds_1973_matching}):
\end{enumerate*}
\begin{equation}\label{eq:domQJoin}
\TJdom{Q}\coloneqq \left\{
x\in \mathbb{R}^E_{\geq 0} \;\middle\vert\; x(\delta(C)) \geq 1 \quad 
\text{for all $Q$-cuts}
\;\; C\subseteq V
\right\}\enspace,
\end{equation}
where a \emph{$Q$-cut} is a set $C\subseteq V$ with $|C\cap Q|$ odd.\footnote{The dominant $\TJdom{Q}$ of the $Q$-join polytope are all points $x\in \mathbb{R}^E$ such that there exists a convex combination $y=\sum_{i=1}^k \lambda_i \chi^{J_i}$ of characteristic vectors $\chi^{J_i}\in \{0,1\}^E$ of $Q$-joins $J_i\subseteq E$ for $i\in \{1,\ldots, k\}$ such that $y\leq x$.}
This readily allows for analyzing Christofides' algorithm in terms of the value of an optimal solution $x^*$ to the Held-Karp relaxation: The minimum length spanning tree $T$ fulfills $\ell(T) \leq \ell(x^*)$, and the shortest $\odd(T)$-join $J\subseteq E$, as computed in Christofides' algorithm, has length at most $\ell(J) \leq \sfrac{\ell(x^*)}{2}$. This analysis, which is due to Wolsey~\cite{wolsey_1980_heuristic}, shows that Christofides' algorithm returns a solution of length no more than $\sfrac{3}{2}\cdot \ell(x^*)$, which, apart from providing an alternative proof of the $\sfrac{3}{2}$-approximation guarantee, also implies an upper bound of $\sfrac{3}{2}$ on the integrality gap of the Held-Karp relaxation for TSP. To date, this remains the best known upper bound on the integrality gap, which is widely believed to be equal to a known lower bound of $\sfrac{4}{3}$.
A crucial element in the above analysis, which turns out to fail in the Path TSP case, is that for any solution $y$ to the Held-Karp relaxation, $\sfrac{y}{2} \in \TJdom{Q}$ for \emph{any} $Q\subseteq V$ of even cardinality. This readily follows from the fact that the Held-Karp relaxation requires $y$ to have a load $y(\delta(C))$ of at least $2$ on \emph{each} cut $C$, whereas $\TJdom{Q}$ only requires a load of at least $1$ on a subset of the cuts (that depends on $Q$).

However, this reasoning does not carry over to the Held-Karp relaxation for Path TSP, which is:
\begin{equation}\label{eq:HKforSTTSP}\tag{Held-Karp relaxation for Path TSP}
\renewcommand\arraystretch{1.2}
\begin{array}{r@{\;}c@{\;}ll}
\multicolumn{1}{c}{\min \ell(x)} & & & \\
     x(\delta(C)) &\geq & 2 & \forall\; C\subsetneq V,\; C\neq \emptyset,\; |C\cap \{s,t\} | \in \{0,2\}\\
     x(\delta(C)) &\geq & 1 & \forall\; C\subseteq V,\; |C\cap \{s,t\} | = 1\\
     x(\delta(v)) &= & 2 & \forall\; v\in V\setminus \{s,t\}\\
     x(\delta(s)) = x(\delta(t)) &= & 1 & \\
     x &\in &\mathbb{R}^E_{\geq 0}\enspace.&
\end{array}
\smallskip
\end{equation}
One can easily show that any feasible solution to the Held-Karp relaxation for Path TSP is in the spanning tree polytope of $G$, which guarantees that the length of a shortest spanning tree $T$ is at most the optimal value of the relaxation. However, a solution $y$ to the above linear program only needs to have a load of at least $1$ on $s$-$t$ cuts, and hence, $\sfrac{y}{2}$ may violate some constraints of $\TJdom{(\odd(T)\symdiff \{s,t\})}$ corresponding to $s$-$t$ cuts $C$.

\subsection{Further discussion on prior results and motivation}

The above explanations also allow for providing some additional context regarding the prior results listed in Table~\ref{tab:prevApprox}. In particular, An, Kleinberg, and Shmoys~\cite{an_2015_improving} showed that, for any solution $y$ to the Held-Karp relaxation for Path TSP, the narrow $s$-$t$ cuts---which are the ones on which $y$ has a load of strictly less than $2$---form a chain. Together with a novel idea in this context of using a distribution over trees in the first step of Christofides' algorithm, obtained through a convex decomposition of $x^*$ into spanning trees, they showed that parity correction can be achieved at an average cost of at most $\sfrac{(\sqrt{5}-1)\ell(x^*)}{2}$.
Algorithms of this type, where a distribution of trees is used in the first step instead of a single minimum length spanning tree, were later dubbed \emph{Best-of-Many Christofides}, because they can be derandomized by taking the best tree from the distribution.
Following~\cite{an_2015_improving}, better ways to bound the cost of parity correction were derived through a variety of new techniques, all exploiting the chain structure of narrow cuts of an optimal solution to the Held-Karp relaxation.
In particular, Seb\H{o}~\cite{sebo_2013_eight-fifth} presented an elegant stronger analysis of the Best-of-Many Christofides' algorithm suggested in~\cite{an_2015_improving}. Vygen~\cite{vygen_2016_reassembling} showed that further improvement is possible by first reassembling the trees appearing in a convex decomposition of $x^*$ before sampling, to obtain desirable properties for cheaper parity correction. An additional strengthening was achieved by Gottschalk and Vygen~\cite{gottschalk_2016_better}, who derived a beautiful structural result showing that a very well-structured convex decomposition into trees is possible by generalizing a result of Gao~\cite{gao_2015_metric}.\footnote{We point the interested reader to~\cite{schalekamp_2018_layers} for an alternative proof based on algorithmic matroid theory of the main structural result shown in~\cite{gottschalk_2016_better}.}
Moreover, Seb\H{o} and van Zuylen~\cite{sebo_2016_salesman} modified the general connectivity plus parity correction approach by first deleting some edges from the initial spanning tree, thus getting a disconnected graph, and, after parity correction, reconnecting the different connected components if the graph did not get reconnected through parity correction. They obtained a $1.529$-approximation through this method, thus significantly narrowing the gap to the $1.5$-approximation for TSP. 

Most recently, Traub and Vygen~\cite{traub_2018_approaching} almost closed the gap by presenting a $(1.5+\epsilon)$-approximation, for any fixed $\epsilon>0$, by introducing a powerful technique based on constructing the initial tree, and a corresponding Held-Karp solution to bound the cost of parity correction, through a recursive dynamic program. 
At each level of their recursion, they construct a new Held-Karp solution that is good on narrow cuts of the current solution, and will be mixed into the current solution through an appropriately chosen convex combination. The key barrier this approach faces to obtain a $1.5$-approximation is that whenever a new Held-Karp solution gets combined with a current one, then the new one has its own narrow cuts, leading to new narrow cuts when mixing this solution with the current one. This seems to be an inherent limit of this approach to reach an approximation ratio of $1.5$.

An approximation ratio of $1.5$ can be considered as a natural and desirable goal to achieve with a connectivity plus parity correction approach. Indeed, this matches the classical $1.5$-approximation of Christofides for TSP, which is also based on such an approach.
However, to go below $1.5$, this approach for Path TSP will face the same barriers as the TSP problem, because any improvement on $1.5$ for Path TSP would also improve on the more than four decades old $1.5$-approximation of Christofides for TSP.

\subsection{Our results}

Our main result is to close the gap between the best approximation factors for Path TSP and TSP, through a new and simple technique inspired by Traub and Vygen~\cite{traub_2018_approaching}, leveraging their dynamic programming approach in a new way.
\begin{theorem}\label{thm:main}
There is an efficient $1.5$-approximation for Path TSP.
\end{theorem}
A major difference of our approach compared to prior methods is that we do not solely focus on narrow cuts, but consider cuts of load strictly less than $3$.
By moving to larger cuts, we lose the crucial chain property of narrow cuts shown in~\cite{an_2015_improving} that was exploited in all prior approaches. However, we show that dynamic programming, in combination with a seminal result by Karger~(see~\cite{karger_1993_global,karger_2000_minimum}) on the number of almost minimum cuts, is sufficiently versatile to deal with cuts of larger value, which do not exhibit a chain structure. This allows us to avoid recursive calls to the dynamic program, which was necessary in the approach of Traub and Vygen~\cite{traub_2018_approaching}. As a consequence, our approach is substantially simpler than the one in~\cite{traub_2018_approaching}, while providing a better approximation guarantee because the use of larger cuts than narrow ones allows for avoiding the introduction of new narrow cuts as we will show later.

\subsection{Organization of the paper}

In Section~\ref{sec:approach}, we introduce our approach and provide details on the involved techniques. In particular, we show that it suffices to obtain a Held-Karp solution with a well-defined set of properties to obtain a $1.5$-approximation. Section~\ref{sec:dynProg} then shows how such a solution can be found through dynamic programming. The dynamic program we use can be interpreted as a modification and considerable simplification of the one of Traub and Vygen~\cite{traub_2018_approaching}, because we do not have recursive calls. Most of the properties that we need to show for our dynamic program can readily be derived through reasonings similar (and often much simplified) to the ones in~\cite{traub_2018_approaching}. Still, for completeness, and due to the fact that we can simplify various steps, we provide full details of our dynamic program in Section~\ref{sec:dynProg}. 
 \section{Our approach}\label{sec:approach}

We show that one can obtain a $1.5$-approximation for Path TSP by following, on a high level, an analysis similar to Wolsey's classical analysis for TSP. More precisely, we find a spanning tree $T\subseteq E$ of $G=(V,E)$ and a point $z$ that is feasible for the Held-Karp relaxation for Path TSP, and which is only needed for the analysis, such that:
\begin{enumerate}
\item\label{item:lightT} $\ell(T) \leq \OPT$,
\item\label{item:lightz} $\ell(z) \leq \OPT$, and
\item\label{item:zIsJoinGood} $\sfrac{z}{2}\in \TJdom{Q_T}$, where
$Q_T \coloneqq \odd(T)\symdiff \{s,t\}$.
\end{enumerate}
A $1.5$-approximation is then obtained through parity correction of $T$, by adding a shortest $Q_T$-join $J$, and shortcutting.
Indeed, the approximation guarantee of $1.5$ readily follows by following Wolsey's analysis for Christofides' algorithm for TSP: Due to~\ref{item:zIsJoinGood}, the length of the shortest $Q_T$-join $J$ satisfies $\ell(J) \leq \sfrac{\ell(z)}{2}$, and hence, the solution obtained after shortcutting has length at most $\ell(T) + \ell(J) \leq \ell(T) + \sfrac{\ell(z)}{2} \leq \sfrac{3}{2}\cdot \OPT$, where the second inequality follows from~\ref{item:lightT} and~\ref{item:lightz}.

We note that also recent improvements for Path TSP prior to this work can be interpreted as showing a weaker version of the three properties above, where instead of the second condition, a bound of type $\ell(z)\leq (1+c) \cdot \OPT$ is shown for some $c>0$.

A canonical approach to achieve these three properties would be to choose $z$ to be an optimal Held-Karp solution, like in Wolsey's analysis for TSP. This guarantees that~\ref{item:lightz} holds, but would require us to find a tree $T$ fulfilling the first and third property, which seems to be a very difficult task.\footnote{Note that for any $z\in \PHK$, a tree $T$ fulfilling~\ref{item:lightT} and~\ref{item:zIsJoinGood} always exists, as an optimal solution $T$ to Path TSP always fulfills these two properties for any $z\in \PHK$, because, if $T$ is an $s$-$t$ path, then $Q_T=\emptyset$ and therefore $\TJdom{Q_T}=\mathbb{R}^E_{\geq 0}$.} 

Our approach starts with an optimal solution $x^*$ to the path version of the Held-Karp relaxation $\min\{\ell(x) : x\in \PHK\}$, where, for brevity, we denote by $\PHK$ the polytope of all feasible solutions to the relaxation:
\begin{equation*}
\PHK \coloneqq \left\{
x\in \mathbb{R}^E_{\geq 0} \;\middle\vert
\begin{array}{r@{\;}c@{\;}ll}
x(\delta(C)) &\geq &2  & \forall\, C\subsetneq V,\; C\neq \emptyset,\; |C\cap \{s,t\}|\in \{0,2\}\\
x(\delta(C))&\geq &1 & \forall\, C\subseteq V,\; |C\cap \{s,t\}|=1\\
x(\delta(v))&= &2 & \forall\, v\in V\setminus \{s,t\}\\
x(\delta(s)) &= &1\\
x(\delta(t)) &= &1
\end{array}\right\}.
\end{equation*}
We highlight that it is well-known that $\PHK$ is contained in the spanning tree polytope of $G$.

To obtain the three properties listed above, we will later set $z= \sfrac{x^*}{2} + \sfrac{y}{2}$ to be the midpoint between $x^*$ and another Held-Karp solution $y\in \PHK$ with $\ell(y)\leq \OPT$, constructed through a dynamic program that depends on $x^*$. Hence, $z$ being a convex combination of two points in $\PHK$ implies $z\in \PHK$. To better understand the properties we want from $y$, recall that we also need to be able to find a short spanning tree $T$ such that $\sfrac{z}{2}\in \TJdom{Q_T}$.
Now recall the description of $\TJdom{Q}$ for some set $Q\subseteq V$ as stated in~\eqref{eq:domQJoin}, which requires that the load on $Q$-odd cuts is at least $1$.
Because $z\in \PHK$, we have $\sfrac{z(\delta(C))}{2}\geq 1$ for any $C\subsetneq V,\; C\neq\emptyset$ with $|C\cap \{s,t\}|\in \{0,2\}$.
For $s$-$t$ cuts $C\subseteq V$, using $y\in \PHK$ only implies $y(\delta(C))\geq 1$. If, additionally, $x^*(\delta(C))\geq 3$, we can again conclude $\sfrac{z(\delta(C))}{2} \geq 1$.
Hence, no matter what tree $T$ we choose later, $\sfrac{z}{2}$ will not violate any of these constraints corresponding to $\TJdom{Q_T}$.
The only constraints of $\TJdom{Q_T}$ we have to take care of thus correspond to $s$-$t$ cuts of $x^*$-load strictly less than $3$, which we denote by $\mathcal{B}(x^*)$:
\begin{equation*}
\mathcal{B}(x^*) \coloneqq \left\{C\subseteq V \mid s\in C,\; t\not\in C,\; x^*(\delta(C))< 3\right\}\enspace.
\end{equation*}
Here, our approach substantially differs from previous methods, which focus on so-called \emph{narrow} cuts with respect to $x^*$, which are $s$-$t$ cuts of $x^*$-load strictly less than $2$. Focusing on narrow cuts allowed for exploiting a crucial property shown by An, Kleinberg, and Shmoys~\cite{an_2015_improving}, namely that they form a chain, and this was at the heart of previous improvements of the approximability of Path TSP.
As we show, we do not need any particular properties on the cuts in $\mathcal{B}(x^*)$, except for a polynomial bound on their number, to determine a point $y$ leading to the three properties highlighted at the beginning of this section.
More precisely, we want to find a point $y\in \PHK$ minimizing $\ell(y)$ among all $y\in \PHK$ that are \emph{$\mathcal{B}(x^*)$-good}, which we define as follows.
\begin{definition}[$\mathcal{B}$-good Held-Karp solution]
\label{def:Bgood}
Let $\mathcal{B}\subseteq \{C\subseteq V \mid s\in C,\; t\not\in C\}$ be a family of $s$-$t$ cuts in $G$. We say that a point $y\in \PHK$ is \emph{$\mathcal{B}$-good} if for every $B\in \mathcal{B}$ we have either
\begin{enumerate}
\item\label{item:BGoodLarge} $y(\delta(B))\geq 3$, or
\item\label{item:BGoodSmall} $y(\delta(B))=1$, and $y$ is integral on the edges $\delta(B)$.
\end{enumerate}
\end{definition}
Notice that condition~\ref{item:BGoodSmall} is equivalent to the property that there is one edge $f\in \delta(B)$ such that $y(f)=1$ and $y(e)=0$ for all $e\in \delta(B)\setminus \{f\}$.

A simple but crucial observation is that the set of all $\mathcal{B}$-good points in $\PHK$ for any family of $s$-$t$ cuts $\mathcal{B}$ is still a relaxation of the Path TSP problem. 
\begin{observation}\label{obs:BGoodIsRelax}
The characteristic vector $\chi^U$ of any Hamiltonian $s$-$t$ path $U\subseteq E$ is $\mathcal{B}$-good for any family $\mathcal{B}$ of $s$-$t$ cuts.
\end{observation}
This follows immediately from the fact that any $s$-$t$ path crosses any $s$-$t$ cut an odd number of times. Thus, a $\mathcal{B}(x^*)$-good point $y\in \PHK$ minimizing $\ell(y)$ will therefore satisfy $\ell(y)\leq \OPT$, as desired.

We are now ready to provide a description of our $1.5$-approximation, highlighted in Algorithm~\ref{alg:mainAlg} below.
\begin{algorithm2e}
\begin{enumerate}[label=\arabic*.,ref=\arabic*,leftmargin=0.5em,itemsep=0.2em]
\item\label{algitem:getHKsol} Obtain optimal solution $x^*$ to Held-Karp relaxation, i.e., $x^*$ is a minimizer of $\min\{\ell(x) \mid x\in \PHK\}$.

\item\label{algitem:getBGoodPoint} Determine a $\mathcal{B}(x^*)$-good point $y\in \PHK$ minimizing $\ell(y)$.

\item\label{algitem:getT} Compute a shortest spanning tree $T\subseteq \supp(y)$ in the graph $(V,\supp(y))$ w.r.t.~lengths $\ell$.\footnotemark{}

\item\label{algitem:getJ} Compute a minimum length $(\odd(T)\symdiff \{s,t\})$-join $J$ in $G$.

\item Return solution obtained by shortcutting the Eulerian $s$-$t$ trail with characteristic vector $\chi^T + \chi^J$.
\end{enumerate}
\caption{A $1.5$-approximation for Path TSP}
\label{alg:mainAlg}
\end{algorithm2e}

To complete the description of our algorithm, and prove that it can be performed efficiently, we still have to provide an efficient procedure to find a minimum length $\mathcal{B}(x^*)$-good point $y\in \PHK$, as required in step~\ref{algitem:getBGoodPoint} of the algorithm. Before doing so, we show that Algorithm~\ref{alg:mainAlg} indeed returns a $1.5$-approximate solution. As highlighted above, we obtain this result by showing that the tree $T$ computed in the algorithm together with the point $z=\sfrac{x^*}{2} + \sfrac{y}{2}$ fulfill the three properties mentioned at the beginning of this section.
\footnotetext{$\supp(y)\coloneqq \{e\in E\mid y(e) > 0\}$ denotes the support of $y$.}
\begin{theorem}
Algorithm~\ref{alg:mainAlg} returns a $1.5$-approximate solution for Path TSP.
\end{theorem}
\begin{proof}
We start by bounding the length of $T$. Notice first that due to Observation~\ref{obs:BGoodIsRelax}, the characteristic vector of the optimal solution to the Path TSP problem is $\mathcal{B}(x^*)$-good, and hence, the point $y$ computed in step~\ref{algitem:getBGoodPoint} of Algorithm~\ref{alg:mainAlg} fulfills
\begin{equation*}
\ell(y) \leq \OPT\enspace,
\end{equation*}
because $y\in \PHK$ is the shortest $\mathcal{B}(x^*)$-good point in $\PHK$.

Using the well-known property that $\PHK$ is contained in the spanning tree polytope of $G$ (see, e.g.,~\cite{an_2015_improving}), and $y\in \PHK$, we obtain $\ell(T)\leq \ell(y)$ because $T$ is a shortest spanning tree in $(V,\supp(y))$, which, together with $\ell(y)\leq \OPT$, implies
\begin{equation}\label{eq:TIsShort}
\ell(T) \leq \OPT\enspace.
\end{equation}

Let $Q_T \coloneqq \odd(T)\symdiff\{s,t\}$. As discussed, to bound the length of the shortest $Q_T$-join $J$, computed in step~\ref{algitem:getJ} of the algorithm, we show that by defining 
\begin{equation*}
z = \frac{1}{2} x^* + \frac{1}{2} y\enspace,
\end{equation*}
the point $\sfrac{z}{2}$ is in $\TJdom{Q_T}$. This will imply the desired result since we then get 
\begin{equation}\label{eq:JIsShort}
\ell(J) \leq \ell\left(\frac{z}{2}\right) = \frac{1}{4}\ell(x^*) + \frac{1}{4}\ell(y) \leq \frac{1}{2} \OPT\enspace,
\end{equation}
where the second inequality follows from $\ell(x^*)\leq \OPT$, because $x^*$ is an optimal solution to a relaxation of the Path TSP problem (namely the Held-Karp relaxation), and $\ell(y)\leq \OPT$.
Finally,~\eqref{eq:TIsShort} and~\eqref{eq:JIsShort} imply that the $s$-$t$ trail described by $\chi^T + \chi^J$ has length at most $\sfrac{3}{2}\cdot \OPT$, as desired.

Thus, it remains to show $\sfrac{z}{2} \in \TJdom{Q_T}$. As already mentioned in the outline of our approach, we have $z\in \PHK$ because it is the midpoint of two points in $\PHK$. Hence, for any $Q_T$-cut $C\subseteq V$ with $|C\cap \{s,t\}| \in \{0,2\}$, we have $\sfrac{z(\delta(C))}{2}\geq 1$, as required by the description~\eqref{eq:domQJoin} of $\TJdom{Q_T}$. Moreover, for any $Q_T$-cut $C\subseteq V$ with $s\in C, t\notin C$, and $C\not\in \mathcal{B}(x^*)$, we have 
\begin{equation*}
\frac{1}{2} z(\delta(C)) = \frac{1}{4}\big( x^*(\delta(C)) + y(\delta(C)) \big) \geq 1\enspace,
\end{equation*}
due to $x^*(\delta(C))\geq 3$ as $C\not\in \mathcal{B}(x^*)$, and $y(\delta(C))\geq 1$ as $y\in \PHK$. (Notice that we do not have to separately consider any cut $C\subseteq V$ with $t\in C$ and $s\notin C$ because it corresponds to the same constraint in $\TJdom{Q_T}$ as the $s$-$t$ cut $V\setminus C$.)
It remains to consider $Q_T$-cuts $C\subseteq V$ with $C\in \mathcal{B}(x^*)$. Because $y$ is $\mathcal{B}(x^*)$-good, we are in one of the following two situations:
\begin{enumerate*}
\item $y(\delta(C)) \geq 3$, or
\item $y(\delta(C)) =1$ and $y$ is integral on the edges $\delta(C)$.
\end{enumerate*}
In the first case we obtain $\sfrac{z(\delta(C))}{2} \geq 1$ because, apart from $y(\delta(C))\geq 3$, we have $x^*(\delta(C))\geq 1$ as $x^*\in \PHK$. Finally, it turns out that there are no $Q_T$-odd cuts $C\in \mathcal{B}(x^*)$ of the second type due to the following. Since $y(e)=0$ for all edges of $\delta(C)$ except for one, and $T\subseteq \supp(y)$, we must have $|T\cap \delta(C)| = 1$. However, an $s$-$t$ cut $C\subseteq V$ with $|T\cap \delta(C)|$ odd cannot be a $Q_T$-cut because
\begin{align*}
|C\cap \odd(T)| &\equiv \sum_{v\in C} |\delta(v)\cap T| \pmod{2} \\
                &= 2 |\{\{u,v\}\in T \mid u,v \in C\}| + |T\cap \delta(C)| \enspace,
\end{align*}
which implies that $|C\cap \odd(T)|$ is odd, and hence $|C\cap Q_T| = |C\cap (\odd(T)\symdiff\{s,t\})|$ is even because $s\in C$ and $t\not\in C$. Thus, $\sfrac{z}{2}\in \TJdom{Q_T}$, as desired.
\end{proof}

It remains to show that step~\ref{algitem:getBGoodPoint} can be performed efficiently. To this end, we show that dynamic programming ideas along the lines of Traub and Vygen~\cite{traub_2018_approaching} can be leveraged for this purpose, without recursive calls. We will prove the following theorem in Section~\ref{sec:dynProg}.
\begin{theorem}\label{thm:dynProg}
Let $\mathcal{B} \subseteq \{C\subseteq V \mid s\in C,\; t\not\in C\}$. One can determine in time polynomial in $|\mathcal{B}|$ and the input size of $(G,s,t,\ell)$ a $\mathcal{B}$-good point $y\in \PHK$ of minimum length $\ell(y)$.
\end{theorem}
Observe that instead of requiring any particular structure of $\mathcal{B}$ (like being a chain), all we need is that $|\mathcal{B}|$ is polynomially bounded. This follows from a seminal result of Karger on near-minimum cuts.
\begin{lemma}\label{lem:cutBound}
Let $z\in \PHK$. Then the family $\mathcal{B}(z)$ of $s$-$t$ cuts of $z$-value strictly less than $3$, i.e.,
\begin{equation*}
\mathcal{B}(z) \coloneqq \{C\subseteq V \mid s\in C,\; t\not\in C,\; z(\delta(C)) < 3\}\enspace,
\end{equation*}
satisfies $|\mathcal{B}(z)|\leq n^4$ and can be computed deterministically in $O(m n^4)$ time, or through a randomized algorithm in $O(n^4 \log^2 n)$ time with high probability, where $n\coloneqq |V|$ and $m\coloneqq |\supp(z)|$.
\end{lemma}
We highlight that a randomized algorithm computing $\mathcal{B}(z)$ in time $O(n^4 \log^2 n)$ with high probability, as mentioned in the above lemma, is an algorithm that, for any $c>0$ (possibly depending on the input), returns in time $O(c\cdot n^4\log^2 n)$ the set $\mathcal{B}(z)$ with probability at least $1-O(n^{-c})$; and with probability at most $O(n^{-c})$ an incorrect family may be returned. Moreover, notice that if $z$ is a vertex of $\PHK$, then $|\supp(z)| = O(n)$, which follows by standard combinatorial uncrossing arguments (see, e.g.,~\cite{vempala_1999_convex}, for an application of this technique to TSP, which easily carries over to Path TSP).\footnote{To show sparsity of vertex solutions, one can also observe that $\PHK$ is the natural relaxation of a degree-bounded spanning tree problem~(see, e.g.,~\cite{schalekamp_2018_layers}), for which Goemans~\cite{goemans_2006_minimum} showed through combinatorial uncrossing that any vertex solution has support bounded by $2n-3$, where $n$ is the number of vertices.}
\begin{proof}[Proof of Lemma~\ref{lem:cutBound}]
Let $H=(V,F)$ be the graph obtained from $(V,\supp(z))$ by adding an additional edge $f=\{s,t\}$ between $s$ and $t$. Let $z_H\in \mathbb{R}^F$ be the vector defined by $z_H(e) = z(e)$ for $e\in \supp(z)$, and $z_H(f) =1 $. Notice that by adding $f$ with a $z_H$-value of $1$, any cut $C\subsetneq V$, $C\neq \emptyset$ in $H$ satisfies $z_H(\delta(C))\geq 2$: For cuts $C$ not separating $s$ and $t$ we have $z_H(\delta(C)) = z(\delta(C))\geq 2$ because $z\in \PHK$, and for cuts separating $s$ and $t$ we have $z_H(\delta_H(C)) = z(\delta(C)) + 1 \geq 2$, again using $z(\delta(C))\geq 1$ as $z\in \PHK$.\footnote{We denote by $\delta_H(C)$ all edges in $H$ with precisely one endpoint in $C$.}
Moreover, since the value of all $s$-$t$ cuts in $H$ with respect to $z_H$ increased by precisely one unit compared to $s$-$t$ cuts in $G$ with respect to $z$, the set $\mathcal{B}(z)$ can be described by
\begin{equation*}
\mathcal{B}(z) = \left\{ C\subseteq V \mid s\in C,\; t\not\in C,\; z_H(\delta_H(C)) < 4\right\}\enspace.
\end{equation*}
Hence, all cuts in $\mathcal{B}(z)$ are cuts of $z_H$-value at most twice the smallest $z_H$-value of any cut in $H$. By a seminal result of Karger~\cite{karger_1993_global,karger_2000_minimum} there are at most $n^4$ such cuts. Moreover, they can be enumerated by a deterministic procedure in $O(|F| n^4) = O(m n^4)$ time~\cite{nagamochi_1997_computing}, or by a randomized one in $O(n^4 \log^2 n)$ time~\cite{karger_1996_new}.
Thus, enumerating all cuts of value strictly less than $4$ and keeping the $s$-$t$ cuts among them, we obtain the claimed guarantees for computing $\mathcal{B}(z)$.
\end{proof}

The above statements imply that Algorithm~\ref{alg:mainAlg} is efficient.
\begin{corollary}
Algorithm~\ref{alg:mainAlg} is efficient.
\end{corollary}
\begin{proof}
It is well-known that finding an optimal Held-Karp solution in step~\ref{algitem:getHKsol} can be done efficiently by using for example the ellipsoid algorithm, or through a compact extended formulation.
The second step of the algorithm is efficient due to Theorem~\ref{thm:dynProg} and Lemma~\ref{lem:cutBound}. Finally, computing a shortest spanning tree in step~\ref{algitem:getT} and finding a shortest $(\odd(T)\symdiff\{s,t\})$-join in step~\ref{algitem:getJ} can be done efficiently through classical procedures (see, e.g.,~\cite{schrijver_2003_combinatorial}).
\end{proof}

It remains to show Theorem~\ref{thm:dynProg}, which we do in the next section.
 \section{Dynamic program to compute shortest $\mathcal{B}$-good solution}\label{sec:dynProg}

Our dynamic program to prove Theorem~\ref{thm:dynProg} can be interpreted as a slight adaptation of the one introduced by Traub and Vygen~\cite{traub_2018_approaching}, avoiding recursive calls. Consequently, one could use similar arguments to the ones in~\cite{traub_2018_approaching} to show the properties we need. However, by avoiding recursive calls, the dynamic program can be presented in a substantially simpler way, and we provide a full description with complete and simplified proofs in the following.

Recall that we are given a family of $s$-$t$ cuts $\mathcal{B}\subseteq \{C\subseteq V \mid s\in C,\; t\not\in C\}$ and our goal is to find a shortest $\mathcal{B}$-good point $y\in \PHK$.
As in~\cite{traub_2018_approaching}, the goal of the dynamic program is to decide in which of the cuts in $\mathcal{B}$ only a single edge should be used. To build up intuition for the dynamic programming approach, consider a $\mathcal{B}$-good point $y\in \PHK$. Let $B_1,\ldots, B_k\in \mathcal{B}$ be the $s$-$t$ cuts in $\mathcal{B}$ such that for each $i\in \{1,\ldots, k\}$, we have $y(\{v_i,u_i\})=1$ for one edge $\{v_i,u_i\} \in \delta(B_i)$ and $y(e)=0$ for all other $e\in \delta(B_i)\setminus \{\{v_i,u_i\}\}$. It is not hard to see that the cuts $B_1, \ldots B_k$ must form a chain, and we choose the numbering such that $B_1\subsetneq B_2 \subsetneq \ldots \subsetneq B_k$.\footnote{That $B_1,\ldots, B_{k}$ form a chain also follows from the fact that narrow cuts with respect to any point $y\in\PHK$ form a chain~\cite{an_2015_improving}, and can be shown by a reasoning analogous to the one used in~\cite{an_2015_improving} as follows. If there were two sets $B_i$ and $B_j$ such that neither $B_i\subseteq B_j$ nor $B_j\subseteq B_i$, then the following contradiction arises: $2=y(\delta(B_i)) + y(\delta(B_j)) \geq y(\delta(B_i\setminus B_j)) + y(\delta(B_j\setminus B_j)) \geq 4$, where the first inequality follows from submodularity and symmetry of the cut function $S \longmapsto y(\delta(S))$, and the second one from $y\in \PHK$ and the fact that both $B_i\setminus B_j$ and $B_j\setminus B_i$ are nonempty sets not containing $s$ or $t$.} Moreover, we name the endpoints of $\{v_i,u_i\}$ such that $v_i\in B_i$ and $u_i\not\in B_i$. See Figure~\ref{fig:dynProgIdea} for an illustration of the introduced terms. For notational convenience, we define $B_0\coloneqq \emptyset$, $B_{k+1}\coloneqq V$, $u_0\coloneqq s$, and $v_{k+1}\coloneqq t$.
Moreover, we remark that we may have $u_i=v_{i+1}$ for some indices $i\in \{0,\ldots, k\}$. In particular, because $\PHK$ imposes degree constraints of $1$ at $s$ and $t$, and because $y$ is $\mathcal{B}$-good, there must be one edge incident to each $s$ and $t$ with $y$-value of $1$, and all other edges incident to $s$ or $t$ have $y$-value of $0$. Hence, $B_1=\{s\}$ and $B_k=V\setminus \{t\}$, and thus $u_0=v_1=s$ and $u_k=v_{k+1}=t$.

A crucial property that is exploited by the dynamic program, and implied by results we show later, is the following. Assume that we knew the cuts $B_1,\ldots, B_{k}$ and edges $\{v_i,u_i\}$ for $i\in \{1,\ldots, k\}$ for a shortest $\mathcal{B}$-good point $y^*\in \PHK$, and we are looking for the shortest $\mathcal{B}$-good point $y\in\PHK$ such that, among the cuts in $\mathcal{B}$, the cuts $B_1,\ldots, B_k$ are precisely the integral $1$-cuts, i.e., cuts $B\in \mathcal{B}$ in which $y$ is integral and $y(\delta(B))=1$, and furthermore, $y(\{v_i,u_i\})=1$ for $i\in \{1,\ldots, k\}$. Then $\mathcal{B}$-good points $y\in \PHK$ of this form are precisely the points $y\in \PHK$ such that for each $i\in \{1,\ldots, k\}$:
\begin{enumerate}
\item $y(\{v_i,u_i\})=1$ and $y(e)=0$ for all $e\in \delta(B_i)\setminus \{\{v_i,u_i\}\}$, and
\item the restriction of $y$ to the subgraph of $G$ induced by $B_{i+1}\setminus B_i$ is a solution to the Held-Karp relaxation for $u_i$-$v_{i+1}$ Path TSP for this graph with the additional property that it has a $y$-load of at least $3$ on each cut $B\in \mathcal{B}$ with $B_i\cup \{u_i\}\subseteq B \subseteq B_{i+1}\setminus \{v_{i+1}\}$.
\end{enumerate}
Because we may have $u_i=v_{i+1}$, we need to be clear about what we mean with a Held-Karp relaxation for Path TSP if start and endpoint are the same. For this case we use the natural extension in our context. Namely, if start and endpoint are equal, then the Held-Karp relaxation for Path TSP is empty if the graph contains at least two vertices, and it consists of only the zero vector if the graph has a single vertex, in which case this vertex must simultaneously be the start and endpoint.

\begin{figure}[h]
\begin{center}
\begin{tikzpicture}
\small
\def\ch{3} 
\def\dx{2.5} 
\def\ns{0.4cm} 
\def\dys{-0.6cm} 
\def\dyw{0.45cm} 

\begin{scope}[dashed,bend right=20]
\draw (0.3*\dx,0) node[below] {$B_0=\emptyset$} to coordinate (m0) ++(0,\ch);
\draw (1*\dx,0) node[below] {$B_1$} to coordinate (m1) ++(0,\ch);

\draw (2.5*\dx,0) node[below] {$B_i$} to coordinate (mi) ++(0,\ch);
\draw (3.5*\dx,0) node[below] {$B_{i+1}$} to coordinate (mi1) ++(0,\ch);

\draw (5*\dx,0) node[below] {$B_k$} to coordinate (mk) ++(0,\ch);
\draw (5.7*\dx,0) node[below] {$B_{k+1}=V$} to coordinate (mk1)++(0,\ch);
\end{scope}

\coordinate (d1) at ($(m1)!0.5!(mi) + (0,\dys)$);
\coordinate (d2) at ($(mi1)!0.5!(mk) + (0,\dys)$);

\begin{scope}[line cap=round,line width=2pt,dash pattern=on 0pt off 3\pgflinewidth]
\draw ($(d1)+(-0.5*\dyw,0)$) -- ($(d1) + (0.5*\dyw,0)$);
\draw ($(d2)+(-0.5*\dyw,0)$) -- ($(d2) + (0.5*\dyw,0)$);
\end{scope}

\begin{scope}[every node/.style={fill=black,circle,inner sep=0em, minimum size=3pt}]
\node (u0) at ($(m0) + (2*\ns,0)$) {};
\node (u1) at ($(m1) + (\ns,0)$) {};
\node (vi) at ($(mi) + (-\ns,0)$) {};
\node (ui) at ($(mi) + (\ns,0)$) {};
\node (vi1) at ($(mi1) + (-\ns,0)$) {};
\node (ui1) at ($(mi1) + (\ns,0)$) {};
\node (vk) at ($(mk) + (-\ns,0)$) {};
\node (vk1) at ($(mk1) + (-2*\ns,0)$) {};
\end{scope}

\begin{scope}[every node/.style={below=2pt}]
\node[align=left,xshift=-0.6em] at ($(u0)+(0.3,0)$) {$s=u_0$\\$\phantom{s}=v_1$};
\node at (u1) {$u_1$};
\node at (vi) {$v_i$};
\node at (ui) {$u_i$};
\node at (vi1) {$v_{i+1}$};
\node at (ui1) {$u_{i+1}$};
\node at (vk) {$v_{k}$};
\node[align=left,xshift=1.0em,yshift=0.2em] at ($(vk1)+(-0.3,0)$) {$t=u_{k}$\\$\phantom{t}=v_{k+1}$};
\end{scope}

\begin{scope}[thick, bend left]
\draw (u0) to (u1);
\draw (vi) to (ui);
\draw (vi1) to (ui1);
\draw (vk) to (vk1);
\end{scope}

\end{tikzpicture} \end{center}

\caption{Illustration of used terminology for a $\mathcal{B}$-good point $y\in \PHK$. The chain $B_1,\ldots, B_k\in \mathcal{B}$ are all $s$-$t$ cuts $B\in \mathcal{B}$ in which $y$ has a value of $1$ on precisely one edge $\{v_i,u_i\}\in \delta(B)$, and value $0$ on all other edges of $\delta(B)$. The naming of the endpoints of $\{v_i,u_i\}$ is chosen such that $v_i\in B_i$ and $u_i\not\in B_i$. For some indices $i\in \{0,\ldots, k\}$ we may have $u_i=v_{i+1}$. In particular, because $\PHK$ imposes degree constraints of $1$ at $s$ and $t$, we have $u_0=v_1=s$ and $u_k=v_{k+1}=t$.
}\label{fig:dynProgIdea}
\end{figure}
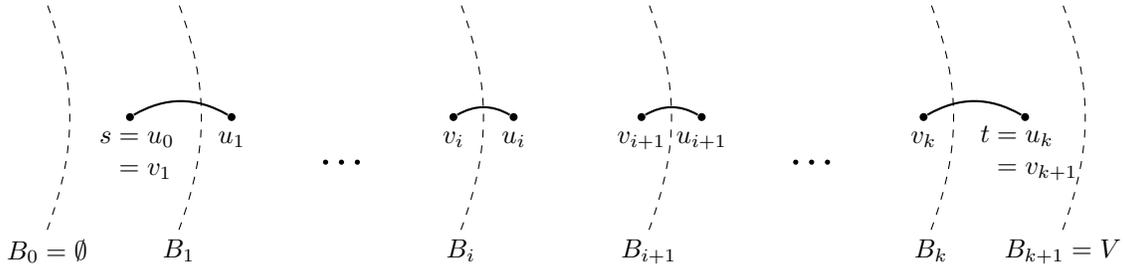

The dynamic program aims at finding the cuts $B_1,\ldots, B_{k}$ and the vertices $v_1,u_1,\ldots, v_k,u_k$, exploiting the above properties.
We present the dynamic program in terms of a shortest path problem in an auxiliary directed graph $H=(N,A)$, with nodes $N$ and arcs $A$, and with length function $d\colon A\longrightarrow \mathbb{R}_{\geq 0}$.
We call elements of $N$ \emph{nodes} and elements of $A$ \emph{arcs} to better separate, in terms of terminology, the auxiliary graph $H$ from the original graph $G$, in which we talk about vertices and edges, respectively.

The node set $N$ of $H$ is defined by
\begin{align*}
N &= N^{\texttt{+}} \cup N^{\texttt{-}}\enspace, \text{where}\\
N^{\texttt{+}} &= \left\{(B,u)\in \mathcal{B} \times V \mid u \not\in B\right\}
  \cup \{(\emptyset, s)\}\enspace, \text{and}\\
N^{\texttt{-}} &= \left\{(B,v)\in \mathcal{B} \times V \mid v \in B\right\}
  \cup \{(V, t)\}\enspace.
\end{align*}
Moreover, the arc set $A$ is given by 
\begin{align*}
A &= A_{\HK} \cup A_E\enspace, \text{where}\\
A_{\HK} &= \left\{\big((B^{\texttt{+}},u), (B^{\texttt{-}},v) \big)\in N^{\texttt{+}} \times N^{\texttt{-}}
\;\middle\vert\; B^{\texttt{+}}\subseteq B^{\texttt{-}}, u,v \in B^{\texttt{-}}\setminus B^{\texttt{+}}\right\}\enspace, and\\
A_{E} &= \left\{\big((B^{\texttt{-}},v), (B^{\texttt{+}},u) \big)\in N^{\texttt{-}} \times N^{\texttt{+}}
\;\middle\vert\; B^{\texttt{-}} = B^{\texttt{+}} \right\}\enspace.
\end{align*}
The lengths $d\colon A\longrightarrow \mathbb{R}_{\geq 0}$ in $H$ are defined as follows. For arcs $a\in A_E$ we set
\begin{equation*}
d(a) = \ell(\{v,u\}) \quad \forall\, a=\big((B,v),(B,u)\big)\in A_E\enspace.
\end{equation*}
The length of an arc $a=((B^{\texttt{+}},u),(B^{\texttt{-}},v))\in A_{\HK}$ is defined to be the optimal value of the following linear program~\ref{eq:LPa}:
\begin{equation}\label{eq:LPa}
\begin{array}{r@{\;}ll}
d(a) = \min &\ell(x) & \\
&x \in \PHK(B^{\texttt{-}}\setminus B^{\texttt{+}}, u, v) & \\
&x(\delta(B)) \geq 3 &\forall B\in \mathcal{B} \text{ with } B^{\texttt{+}}\subseteq B \subseteq B^{\texttt{-}} \text{ and } 
u\in B, v\not\in B\enspace,
\end{array}
\tag{LP($a$)}
\end{equation}
where, for any vertex set $W\subseteq V$ and vertices $u,v\in W$, the polytope $\PHK(W,u,v)\subseteq \mathbb{R}^E$ describes the Held-Karp relaxation for $u$-$v$ Path TSP in $G[W]$, the subgraph of $G$ induced by the vertex set $W$.
For convenience, we define $\PHK(W,u,v)$ to be a polytope in $\mathbb{R}^E$ (instead of $\mathbb{R}^{E[W]}$, where $E[W]\subseteq E$ are all edges with both endpoints in $W$), where all coordinates corresponding to edges $e\in E\setminus E[W]$ are $0$ for any point in $\PHK(W,u,v)$.
More formally, for $u\neq v$ we have
\begin{equation*}
\PHK(W,u,v) \coloneqq \left\{
x\in \mathbb{R}^E_{\geq 0} \;\middle\vert\;
\begin{array}{r@{\;}c@{\;}ll}
x(\delta(C)) &\geq &2  & \forall\, C\subsetneq W,\; C\neq\emptyset,\; |C\cap \{u,v\}|\in \{0,2\}\\
x(\delta(C)) &\geq &1 & \forall\, C\subseteq W,\; |C\cap \{u,v\}|=1\\
x(\delta(w)) &= &2 & \forall\, w\in W\setminus \{u,v\}\\
x(\delta(u)) &= &1\\
x(\delta(v)) &= &1\\
x(e) &= &0 &\forall\, e\in E\setminus E[W]
\end{array}\right\}.
\end{equation*}
Moreover, if $u=v$ and $W=\{u\}$, then $\PHK(W,u,v)$ is defined to only contain the zero vector in $\mathbb{R}^E$, which corresponds to the trivial $u$-$u$ path not containing any edges. Finally, for $u=v$ and $|W|\geq 2$, the polytope $\PHK(W,u,v)$ is empty, and the value of $d(a)$ is set to $\infty$. (Equivalently, we could also simply delete any such arc $a$ from the auxiliary graph because it will be irrelevant in the shortest path problem we solve next.)

To find a $\mathcal{B}$-good point $y\in \PHK$ minimizing $\ell(y)$, we compute a shortest $(\emptyset,s)$-$(V,t)$ path with respect to $d$ in $H$. Lemma~\ref{lem:yIsMinimizer} below together with Observation~\ref{obs:BGoodIsRelax} imply that this path has finite length.
Let $(\emptyset,s)$, $(B_1,v_1)$, $(B_1,u_1)$, $(B_2,v_2)$, \ldots, $(B_k,u_k)$, $(V,t)\in N$ be the nodes on this shortest path. 
For convenience, we define $B_0\coloneqq\emptyset$, $B_{k+1}\coloneqq V$, $u_0 \coloneqq s$, and $v_{k+1}\coloneqq t$. Notice that by construction of $H$, we have $B_0 \subsetneq B_1 \subsetneq \ldots \subsetneq B_{k+1}$.
For each $i\in \{0,\ldots, k\}$, let $x^i\in \mathbb{R}^E$ be an optimal solution to~\ref{eq:LPa} for $a=((B_i,u_i),(B_{i+1},v_{i+1}))$. Note that the linear program~\ref{eq:LPa} is not infeasible, and $x^i$ is therefore well-defined, because the length of the shortest $(\emptyset,s)$-$(V,t)$ path is finite. We set
\begin{equation}\label{eq:defy}
y \coloneqq \sum_{i=0}^k x^i + \sum_{i=1}^k \chi^{\{\{v_i,u_i\}\}}\enspace.
\end{equation}
We highlight that by definition of the lengths $d$ in $H$, we have that $\ell(y)$ is equal to the length $d^*$ of a shortest $(\emptyset,s)$-$(V,t)$ path in $H$ with respect to $d$.

It remains to show that $y$ is indeed a shortest $\mathcal{B}$-good point, which we show below with the next three lemmas.
To prove these statements, it is convenient to use a well-known alternative description of $\PHK$ in terms of a degree-bounded spanning tree relaxation (see, e.g.,~\cite{schalekamp_2018_layers}), namely
\begin{equation}\label{eq:PHKasDBST}
\PHK = \left\{ x\in \mathbb{R}^E_{\geq 0} \;\middle\vert\;
\begin{array}{ll}
x \in \PST & \\
x(\delta(v))=2 & \forall\, v\in V\setminus \{s,t\}\\
x(\delta(v)) = 1 & \forall v\in \{s\}\symdiff \{t\}
\end{array}
\right\},
\end{equation}
where $\PST$ is the spanning tree polytope of $G$, which, as shown by Edmonds~\cite{edmonds_1971_matroids}, can be described by
\begin{equation*}
\PST = \left\{x\in \mathbb{R}_{\geq 0}^E \;\middle\vert\;
\begin{array}{r@{\;}c@{\;}ll}
x(E) &= &|V|-1\\
x(E[W]) &\leq &|W|-1 &\forall\; W\subseteq V, \;W\neq \emptyset 
\end{array}
\right\}.
\end{equation*}
We highlight that in~\eqref{eq:PHKasDBST} we allow for $s=t$, which will be useful later on when using this description for Held-Karp Path TSP relaxations $\PHK(W,u,v)$ for subroblems with $u=v$, used for the computation of the distances $d$.

\begin{lemma}\label{lem:yIsMinimizer}
The length $d^*$ of a shortest $(\emptyset,s)$-$(V,t)$ path in $H$ with respect to $d$ satisfies $d^* \leq \min\{\ell(z) \mid z\in \PHK, z \text{ is $\mathcal{B}$-good}\}$.
\end{lemma}
\begin{proof}
Let $z\in \PHK$ be a $\mathcal{B}$-good point, and we will show that $d^*\leq \ell(z)$. Let $\mathcal{B}_z\subseteq \mathcal{B}$ be the family of all $B\in \mathcal{B}$ such that $z(f)=1$ for precisely one edge $f\in \delta(B)$ and $z(e)=0$ for all $e\in \delta(B)\setminus \{f\}$. Hence, these are the sets in $\mathcal{B}$ for which $z$ fulfills point~\ref{item:BGoodSmall} of the definition of being $\mathcal{B}$-good.
Again, one can easily check that $\mathcal{B}_z$ is a chain, which also follows by the fact shown in~\cite{an_2015_improving} that narrow cuts with respect to any point in $\PHK$ form a chain. Hence, $\mathcal{B}_z=\{B_1, \ldots B_k\}$ with $s\in B_1\subsetneq B_2 \subsetneq \ldots \subsetneq B_k \not\ni t$. Moreover, we set $B_0\coloneqq \emptyset$ and $B_{k+1}\coloneqq V$. For $i\in \{1,\ldots, k\}$, let $\{v_i,u_i\}$ be the one edge in $\delta(B_i)$ with $z(\{v_i,u_i\})=1$, and we choose the naming such that $v_i \in B_i$ and $u_i \not\in B_i$. Furthermore, we set $u_0\coloneqq s$ and $v_{k+1}\coloneqq t$.

To show $d^*\leq \ell(z)$, we show that in $H$, the length $\overline{d}$ of the path along the vertices $(B_0,u_0)$, $(B_1,v_1)$, $(B_1,u_1)$, $(B_2,v_2), \dots , (B_{k+1},v_{k+1})$ is no larger than $\ell(z)$. The inequality $d^* \leq \ell(z)$ then follows because $d^*$ is the length (with respect to $d$) of a shortest $(\emptyset,s)$-$(V,t)$ path in $H$, and thus, $d^*\leq \overline{d}$.
Hence, it suffices to show that for each $i\in \{0,\ldots, k\}$, the vector $z^i\in \mathbb{R}^E$, defined to be equal to $z$ on all edges in $E[B_{i+1}\setminus B_i]$ and $z(e)=0$ for $e\in E\setminus E[B_{i+1}\setminus B_i]$, is feasible for~\ref{eq:LPa} with $a=((B_i,u_i),(B_{i+1},v_{i+1}))$. Indeed, this will imply $\ell(x^i) \leq \ell(z^i)$, where $x^i$ is an optimal solution to this LP, and hence
\begin{equation*}
\ell(z) = \sum_{i=0}^k \ell(z^i) + \sum_{i=1}^k \ell(\{v_i,u_i\})
        \geq \sum_{i=0}^k \ell(x^i) + \sum_{i=1}^k \ell(\{v_i,u_i\})
        = \overline{d}
        \geq d^*\enspace.
\end{equation*}
Because $z$ is $\mathcal{B}$-good, we have $z^i(\delta(B)) = z(\delta(B)) \geq 3$ for all $B\in \mathcal{B}$ with $B_i \subsetneq B \subsetneq B_{i+1}$ and $u_i\in B, v_{i+1}\notin B$, as required by~\ref{eq:LPa}. Thus, it remains to observe that $z^i\in \PHK(B_{i+1}\setminus B_i, u_i, v_{i+1})$, i.e., $z^i$ is a solution to the Held-Karp relaxation for $u_i$-$v_{i+1}$ Path TSP on the graph $G[B_{i+1}\setminus B_i]$.
This readily follows from the fact that $z\in \PHK$. In particular, one can check that $z^i$ satisfies the description of $\PHK(B_{i+1}\setminus B_i, u_i, v_{i+1})$ given by~\eqref{eq:PHKasDBST}, where the role of $G$, $s$, and $t$ is replaced by $G[B_{i+1}\setminus B_i]$, $u_i$, and $v_{i+1}$, respectively. The fact that $z^i$ fulfills the degree constraints immediately follows from $z\in \PHK$ and by using that $\{v_1,u_1\},\ldots, \{v_k,u_k\}$ all have $z$-value $1$, whereas all other edges in $\delta(B_1)\cup \ldots \cup \delta(B_k)$ have $z$-value $0$.
Moreover, since $z$ is in the spanning tree polytope of $G$, we have that $z^i$, which is obtained from $z$ by setting some coordinates to $0$, is in the forest polytope of $G[B_{i+1}\setminus B_i]$. However, due to the degree constraints, we have $z^i(E[B_{i+1}\setminus B_i])= |B_{i+1}\setminus B_i|-1$, and thus, $z^i$ must be in the spanning tree polytope of $G[B_{i+1}\setminus B_i]$. This shows $z^i\in \PHK(B_{i+1}\setminus B_i, u_i, v_{i+1})$ as desired and finishes the proof.
\end{proof}

It remains to show $y\in \PHK$ and that $y$ is indeed $\mathcal{B}$-good, which then implies, together with Lemma~\ref{lem:yIsMinimizer}, that $y\in \PHK$ is a shortest $\mathcal{B}$-good point, as desired.
In the proofs of the next two lemmas, let $(B_0=\emptyset, u_0=s), (B_1,v_1),(B_1,u_1),\ldots, (B_k,u_k), (B_{k+1}=V,v_{k+1}=t)\in N$ be the nodes on the shortest $(\emptyset,s)$-$(V,t)$ path in $H$ used to define $y$.

\begin{lemma}\label{lem:yIsInPHK}
$y\in \PHK$.
\end{lemma}
\begin{proof}
Recall that for $i\in \{0,\ldots, k\}$, the vector $x^i$ used in the definition of $y$ (see~\eqref{eq:defy}) satisfies $x^i \in \PHK(B_{i+1}\setminus B_i, u_i, v_{i+1})$, i.e., it describes a solution to the Held-Karp relaxation for $u_i$-$v_{i+1}$ Path TSP in the induced subgraph $G[B_{i+1}\setminus B_i]$ of $G$.
\footnote{For simplicity, we sometimes refer to $x^i$ (and other vectors with support within $E[B_{i+1}\setminus B_i]$) as a vector on the edges of $G[B_{i+1}\setminus B_i]$, for example by saying that it is a solution to the Held-Karp relaxation for Path TSP in the graph $G[B_{i+1}\setminus B_i]$; even though $x^i\in \mathbb{R}^E$ instead of $x^i\in \mathbb{R}^{E[B_{i+1}\setminus B_i]}$.}

First, this implies that $y$ fulfills the degree constraints in~\eqref{eq:PHKasDBST}. Indeed, let $v\in V$ and let $i\in \{0,\ldots, k\}$ be such that $v\in B_{i+1}\setminus B_i$. The value $y(\delta(v))$ is equal to $x^i(\delta(v))$ plus one additional unit if
$v=u_{i}$ and $i\neq 0$ (due to the edge $\{v_{i},u_{i}\}$) plus one additional unit if
$v=v_{i+1}$ and $i\neq k$ (due to the edge $\{v_{i+1},u_{i+1}\}$). Together with $x^i \in \PHK(B_{i+1} \setminus B_i, u_i, v_{i+1})$, this implies that $y$ fulfills the degree constraints in~\eqref{eq:PHKasDBST}.

Second, we also have $y\in \PST$ due to the following facts:
\begin{enumerate}
\item\label{item:BPartition} $B_{1}\setminus B_0,\ldots, B_{k+1}\setminus B_k$ is a partition of $V$.
\item\label{item:xiIsInPST} For each $i\in \{0,\ldots, k\}$, the vector $x^i\in \PHK(B_{i+1}\setminus B_{i}, u_i, v_{i+1})$, when restricted to $E[B_{i+1}\setminus B_i]$, is a solution to the Held-Karp relaxation for $u_i$-$v_{i+1}$ Path TSP in $G[B_{i+1}\setminus B_i]$. This implies by~\eqref{eq:PHKasDBST} that it is in the spanning tree polytope of $G[B_{i+1}\setminus B_i]$.
\item\label{item:oneEdgesConnect} The edges $\{v_i,u_i\}$ for $i\in \{1,\dots k\}$, which all have a $y$-value of $1$, are a spanning tree (actually even a path) in the graph obtained from $G$ when contracting $B_{i+1}\setminus B_i$ for each $i\in \{0,\ldots, k\}$.
\end{enumerate}
One way to see that the above points indeed imply $y\in \PST$, is to show the equivalent statement that for any $c\in \mathbb{R}^E$, there is a spanning tree $T \subseteq E$ in $G$ such that $c(T) \leq c(y)$.
Indeed, point~\ref{item:xiIsInPST} implies that, for any $i\in \{0,\ldots, k\}$, there is a spanning tree $T_i$ in $G[B_{i+1}\setminus B_i]$ with $c(T_i)\leq c(x^i)$. Together with~\ref{item:BPartition} and~\ref{item:oneEdgesConnect}, this implies that $T\coloneqq \left(\bigcup_{i=0}^k T_i\right) \cup \left(\bigcup_{i=1}^k \{\{v_i,u_i\}\} \right)$ is a spanning tree in $G$ with $c(T)\leq c(y)$, as desired.

Hence, $y$ satisfies all constraints of~\eqref{eq:PHKasDBST} and thus $y\in \PHK$.
\end{proof}

\begin{lemma}\label{lem:yIsBGood}
$y$ is $\mathcal{B}$-good.
\end{lemma}
\begin{proof}
Note that for any $i\in \{1,\ldots, k\}$, we have by construction of $y$ that $y(\{v_i,u_i\})=1$, and $y(e)=0$ for all $e\in \delta(B_i)\setminus \{\{v_i,u_i\}\}$. Hence, the $s$-$t$ cuts in $\mathcal{B}$ that correspond to cuts $B_i$ for $i\in \{1,\ldots, k\}$ all satisfy point~\ref{item:BGoodSmall} of the definition of $\mathcal{B}$-good, i.e., Definition~\ref{def:Bgood}. We show that any other cut $B\in \mathcal{B} \setminus \{B_1,\ldots, B_k\}$ satisfies point~\ref{item:BGoodLarge} of Definition~\ref{def:Bgood}.
For this we distinguish whether $\{B_1,\ldots, B_k\} \cup \{B\}$ is a chain or not.

If $\{B_1,\ldots, B_k\}\cup \{B\}$ is not a chain, then there is some $i\in \{0,\ldots, k\}$ such that neither $B\subseteq B_i$ nor $B_i \subseteq B$. This implies
\begin{align*}
y(\delta(B)) + 1 &= y(\delta(B)) + y(\delta(B_i))\\
  &\geq y(\delta(B\setminus B_i)) + y(\delta(B_i\setminus B))\\
  &\geq 4\enspace,
\end{align*}
due to the following.
The equality follows from $y(\delta(B_i))=1$ as shown at the beginning of the proof. The first inequality is a well-known property of cut functions $C \longmapsto y(\delta(C))$, following from the fact that they are symmetric and submodular. 
The second inequality is a consequence of $y\in \PHK$, the fact that $B\setminus B_i\neq \emptyset$ and $B_i\setminus B\neq \emptyset$ by choice of $B_i$, and the property that both $B\setminus B_i$ and $B_i\setminus B$ contain neither $s$ nor $t$ because $B$ and $B_i$ are $s$-$t$ cuts.

\medskip

Now assume that $\{B_1,\ldots, B_k\}\cup \{B\}$ is a chain, and hence, there is some $i\in \{0,\ldots, k\}$ such that $B_i\subsetneq B\subsetneq B_{i+1}$. Recall that $x^i$ used in the definition~\eqref{eq:defy} of $y$ is a feasible solution to~\ref{eq:LPa} for $a=((B_i,u_i),(B_{i+1}, v_{i+1}))$, which we exploit in the following.

If $u_i\in B$ and $v_{i+1}\not\in B$, then the constraints of~\ref{eq:LPa} imply $x^i(\delta(B))\geq 3$, and thus $y(\delta(B))\geq x^i(\delta(B)) \geq 3$, because $y\geq x^i$ component-wise.
If $u_i\not\in B$ and $v_{i+1}\in B$, then $\{v_i,u_i\}\in \delta(B)$, $\{v_{i+1},u_{i+1}\} \in \delta(B)$, and, moreover, $x^i(\delta(B))\geq 1$, which again implies $y(\delta(B)) \geq y(\{v_i,u_i\})+y(\{v_{i+1}, u_{i+1}\}) + x^i(\delta(B))\geq 3$. Finally, if $B$ does not separate $u_i$ and $v_{i+1}$, then $x^i(\delta(B)) \geq 2$ because $x^i\in \PHK(B_{i+1}\setminus B_i, u_i, v_{i+1})$, and, moreover, either $\{u_i,v_i\}\in \delta(B)$ (if $u_i,v_{i+1}\not\in B$), or $\{u_{i+1},v_{i+1}\}\in \delta(B)$ (if $u_{i}, v_{i+1} \in B$). In both cases we again obtain $y(\delta(B))\geq 3$, as desired.
\end{proof}

Finally, we briefly discuss the running time of the suggested procedure.

\begin{lemma}\label{lem:yRunningTimeBound}
The running time of the suggested algorithm to compute $y$ is bounded by $O(|\mathcal{B}|^2 n^2\cdot p(n,|\mathcal{B}|,\ell))$, where $p(n,k,\ell)$ is an upper bound on the running time needed to solve a linear program of type~\ref{eq:LPa} with $|\mathcal{B}| \leq k$. (We recall that $n$ is the number of vertices of $G$ and $\ell$ is the linear objective of~\ref{eq:LPa}.)
\end{lemma}
\begin{proof}
The bottleneck of the suggested algorithm consists of solving the linear programs~\ref{eq:LPa} for all arcs $a$ of $H$. The node set $N$ of $H$ is a subset of $\mathcal{B}\times V$. Hence, $|N|\leq |\mathcal{B}| n$ and $H$ has no more than $O(|N|^2) = O(|\mathcal{B}|^2 n^2)$ arcs. Thus, at most $O(|\mathcal{B}|^2 n^2)$ linear programs of type~\ref{eq:LPa} have to be solved, each taking time bounded by $p(n,|\mathcal{B}|,\ell)$.
\end{proof}

Theorem~\ref{thm:dynProg} is now a direct consequence of Lemma~\ref{lem:yIsMinimizer},~\ref{lem:yIsInPHK},~\ref{lem:yIsBGood}, and~\ref{lem:yRunningTimeBound}, and the fact that a linear program of type~\ref{eq:LPa} can be solved efficiently for polynomially bounded $|\mathcal{B}|$, i.e., $p(n,|\mathcal{B}|,\ell)$ can be chosen to be polynomial in $n$, $|\mathcal{B}|$, and the input size of $\ell$. Analogous to the Held-Karp relaxation,~\ref{eq:LPa} can for example be solved via the ellipsoid method, or by using an extended formulation of the Held-Karp relaxation and adding the constraints $x(\delta(B))\geq 3$ for the cuts $B\in \mathcal{B}$ satisfying the properties described in~\ref{eq:LPa}. We notice that~\ref{eq:LPa} can be solved in strongly polynomial time through known techniques (see, e.g., discussion and references in~\cite[Section 58.5]{schrijver_2003_combinatorial}), and hence, $p(n,k,\ell)$ can be chosen to be a polynomial only depending on $n$ and $k$ but not on $\ell$.\footnote{The discussion in~\cite{schrijver_2003_combinatorial} is for the Held-Karp relaxation of TSP and carries over to Path TSP.}
In particular, this implies that Algorithm~\ref{alg:mainAlg} can be implemented in strongly polynomial time, because all other steps also admit strongly polynomial time implementations.

 \section{Conclusions}

We presented a $1.5$-approximation for Path TSP, thus matching the guarantee of the classical and still unbeaten $1.5$-approximation of Christofides for TSP. Because any $\alpha$-approximation for Path TSP implies an $\alpha$-approximation for TSP, any further improvement in terms of approximation factor for Path TSP will also improve the TSP approximation ratio below $1.5$, which is a long-standing open problem in the field.
A key difference of our approach compared to previous ones is that we do not rely on narrow cuts. More precisely, we show that dynamic programming together with Karger's polynomial bound on near-minimum cuts is sufficiently versatile to deal with $s$-$t$ cuts of larger size. This allows us to avoid recursive calls to dynamic programs as done by Traub and Vygen~\cite{traub_2018_approaching}, and therefore leads to a simpler algorithm without additional error term in the approximation factor.
Finally, we highlight that our approach does not imply a new upper bound on the integrality gap of the Held-Karp relaxation for Path TSP, which is believed to be $1.5$. The currently best upper bound on the integrality gap is $1.5248$ by Traub and Vygen~\cite{traub_2018_improved}, which is derived by providing an improved analysis of an approach by Seb\H{o} and van Zuylen~\cite{sebo_2016_salesman}.
 \subsection*{Acknowledgments}

The author is grateful to Martin N\"agele, Andr\'as Seb\H{o}, and the anonymous reviewers for many helpful comments that improved the presentation of the results.
 
\bibliographystyle{plain}

\end{document}